\documentclass[final]{dmtcs-episciences}

\usepackage{amssymb} \usepackage[utf8]{inputenc} \usepackage{enumerate}
\usepackage{nicefrac}

\newtheorem{theorem}{Theorem}
\newtheorem{definition}[theorem]{Definition}
\newtheorem{lemma}[theorem]{Lemma}
\newtheorem{claim}[theorem]{Claim}

\newcommand{\DC}{\textsc{Defective Coloring}} %Problem name
\newcommand{\Part}{\textsc{4-Partition}} % 4-Partition
\newcommand{\D}{\ensuremath\Delta^*} %Deficiency
\newcommand{\C}{\ensuremath\mathrm{\chi_d}} %# of colors
\newcommand{\TW}{\ensuremath\textrm{tw}} %treewidth

\usepackage[round]{natbib}

\author{R{\'e}my Belmonte\affiliationmark{1} \and Michael Lampis\affiliationmark{1} \and Valia Mitsou\affiliationmark{2}}

\affiliation{
Université Paris-Dauphine, PSL Research University, CNRS, UMR 7243, LAMSADE, Paris, France\\
Université de Paris, CNRS, IRIF, UMR 8243, Paris, France
}

\title{Defective Coloring on Classes of Perfect Graphs\thanks{This work has been supported by the ANR-14-CE25-0006 project of the French National Research Agency and by the project GRAPA - Graph Algorithms for Parameterized
Approximation - 38593YJ of PHC Sakura program.}}

\keywords{Defective Coloring, Split Graphs, Cographs}

\received{2018-10-29}
\revised{2020-11-19}
\accepted{2021-12-12}
\publicationdetails{24}{2022}{1}{1}{4926}

\begin{document}

\maketitle

\begin{abstract}

In \DC\ we are given a graph $G$ and two integers $\C,\D$ and are asked if we
can $\C$-color  $G$ so that the maximum degree 
induced by any color class is at most $\D$. We show that this natural
generalization of \textsc{Coloring} is much harder on several
basic graph classes. In particular, we show that it is NP-hard on split
graphs, even when one of the two parameters $\C,\D$ is set to the smallest
possible fixed value that does not trivialize the problem ($\C=2$ or $\D=1$).
We also give a simple treewidth-based DP algorithm which, together with the
aforementioned hardness for split graphs, also completely determines the
complexity of the problem on chordal graphs.

We then consider the case of cographs and show that, somewhat surprisingly,
\DC\ turns out to be one of the few natural problems which are NP-hard on this
class.  We complement this negative result by showing that \DC\ is in P for
cographs if either $\C$ or $\D$ is fixed; that it is in P for trivially perfect
graphs; and that it admits a sub-exponential time algorithm for cographs when
both $\C$ and $\D$ are unbounded.

\end{abstract}

\section{Introduction}

In this paper we study the computational complexity of \DC, which is also known
in the literature as \textsc{Improper Coloring}: given a graph and two
parameters $\C,\D$ we want to color the graph with $\C$ colors so that every
color class induces a graph with maximum degree at most $\D$.  \DC\ is a very
natural generalization of \textsc{Graph Coloring}, which corresponds to the
case $\D=0$. As a result, since the introduction of this problem more than
thirty years ago (\cite{CCW86,andrews1985generalization}) a great
deal of research effort has been devoted to its study. Among the topics that
have been investigated are its extremal properties
(\cite{FrickH94,KimKZ14,KimKZ16,BorodinKY13,AchuthanAS11,GoddardX16}), especially
on planar graphs and graphs on surfaces
(\cite{CowenGJ97,Archdeacon87,ChoiE16,HavetS06}), as well as its asymptotic
behavior on random graphs (\cite{KangM10,KangMS08}).  Lately, the problem has
attracted renewed interest due to its applicability to communication networks,
with the coloring of the graph modeling the assignment of frequencies to nodes
and $\D$ representing some amount of tolerable interference.  This has led to
the study of the problem on Unit Disk Graphs \cite{HavetKS09} as well as
various classes of grids \cite{AraujoBGHMM12,BermondHHS10,ArchettiBHCG15}. Weighted
generalizations 
have also been
considered \cite{Bang-JensenH15,GudmundssonMS16}. More background can be found
in the survey by \cite{frick1993survey} or Kang's PhD thesis
(\cite{Kang08}).

Our main interest in this paper is to study the computational complexity of
\DC\ through the lens of structural graph theory, that is, to investigate the
classes of graphs for which the problem becomes tractable. Since \DC\
generalizes \textsc{Graph Coloring} we immediately know that it is NP-hard
already in a number of restricted graph classes and for small values of
$\C,\D$.  Nevertheless, the fundamental question we would like to pose is what
is the \emph{additional} complexity brought to this problem by the freedom to
produce a slightly improper coloring. In other words, we ask what are the graph
classes where even though \textsc{Graph Coloring} is easy, \DC\ is still hard
(and conversely, what are the classes where both are tractable).  Though some
results of this type are already known (for example 
\cite{CowenGJ97} prove that the problem is hard even on planar graphs for
$\C=2$), this question is not well-understood.  Our focus on this paper is to
study \DC\ on subclasses of perfect graphs, which are perhaps the most widely
studied class of graphs where \textsc{Graph Coloring} is in P.  The status of
the problem appears to be unknown here, and in fact its complexity on interval
and even proper interval graphs is explicitly posed as an open question in
\cite{HavetKS09}.

\begin{table}
\begin{tabular}{ll}

\begin{minipage}{0.63\textwidth}
\centering
\begin{tabular}{|c|c|}
\hline
 \textbf{Chordal graphs} & \textbf{Cographs} \\ 
\hline
 NP-hard on Split if $\C\ge 2$ & NP-hard \\
Theorem \ref{thm:split-colors} & Theorem \ref{thm:cograph-hard} \\
\hline
 NP-hard on Split if $\D\ge 1$ & In P if $\C$ or $\D$ is fixed\\
Theorem \ref{thm:sd2} & Theorems \ref{thm:DP1},\ref{thm:DP2} \\
\hline
 In P if $\C,\D$ fixed & Solvable in $n^{O\left(n^{\nicefrac{4}{5}}\right)}$ \\
Theorem \ref{thm:split-alg} & Theorem \ref{thm:subexp} \\
\hline
\multicolumn{2}{|c|}{In P on Trivially perfect for any $\C,\D$} \\
\multicolumn{2}{|c|}{Theorem \ref{thm:triv-perf}} \\
\hline
\end{tabular}
\smallskip

\end{minipage}&

\begin{minipage}{0.35\textwidth}
\centering
\includegraphics[width=\textwidth]{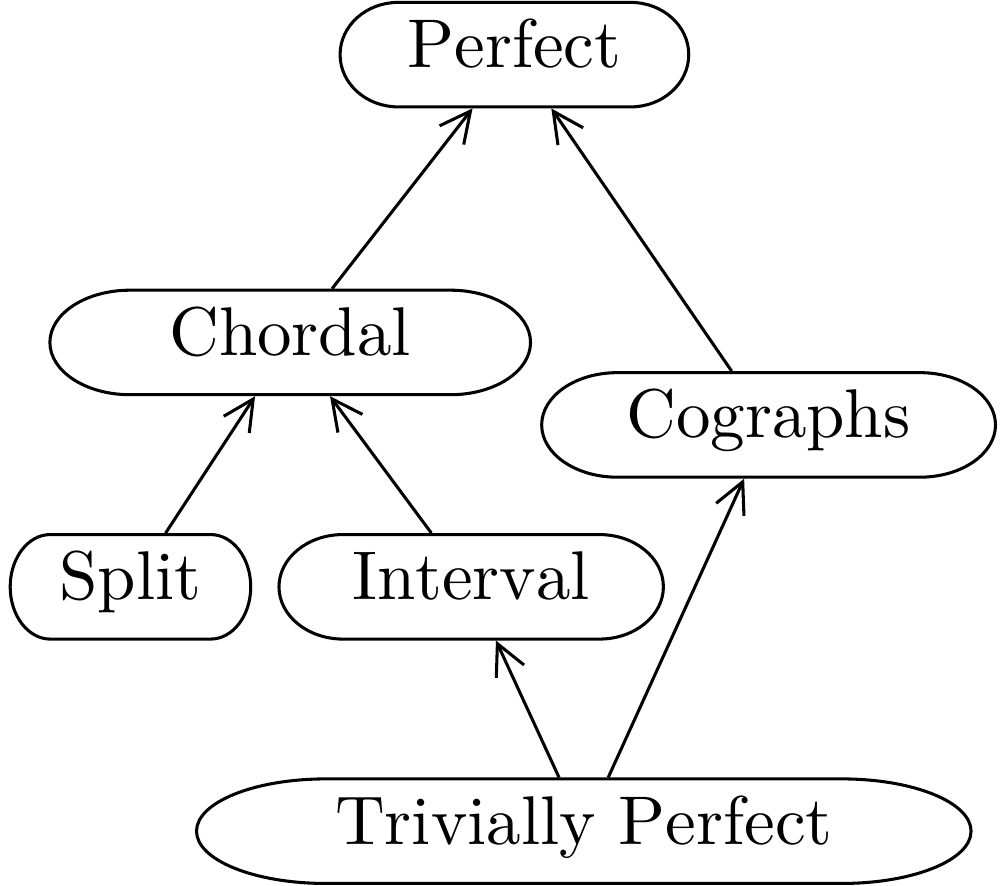}
\end{minipage}

\end{tabular}
\caption{Summary of results} \label{tab:results}
\end{table}

Our results revolve around two widely studied classes of perfect graphs: split
graphs and cographs. For split graphs we show not only that \DC\ is
NP-hard, but that it remains NP-hard even if either $\C$ or $\D$ is a constant
with the smallest possible non-trivial value ($\C\ge 2$ or $\D\ge 1$). To
complement these negative results we provide a treewidth-based DP algorithm
which runs in polynomial time if both $\C$ and $\D$ are
constant, not only for split graphs, but also for chordal graphs. This
generalizes a previous algorithm of \cite{HavetKS09} on interval
graphs.

We then go on to show that \DC\ is also NP-hard when restricted to cographs.
We note that this result is somewhat surprising since relatively few natural
problems are known to be hard for cographs. %\todo{Citation needed?} 
We complement this negative result in several ways. First, we show that \DC\
becomes polynomially solvable on trivially perfect graphs, which form a large
natural subclass of cographs. Second, we show that, unlike the case of split
graphs, \DC\ is in P on cographs if either $\C$ or $\D$ is fixed. Both of these
results are based on dynamic programming algorithms. Finally, by combining the
previous two algorithms with known facts about the relation between $\C$ and
$\D$ we obtain a sub-exponential time algorithm for \DC\ on cographs.  We note
that the existence of such an algorithm for split graphs is ruled out by our
reductions, under the Exponential Time Hypothesis. Table \ref{tab:results} 
summarizes our results. For the reader's convenience, it also
depicts an inclusion diagram for the graph classes that we mention.

\section{Preliminaries and Definitions}

%General

We use standard graph theory terminology, see e.g. \cite{Diestel}.  In
particular, for a graph $G=(V,E)$ and $u\in V$ we use $N(u)$ to denote the set
of neighbors of $u$, $N[u]$ denotes $N(u)\cup\{u\}$, and for $S\subseteq V$ we
use $G[S]$ to denote the subgraph induced by the set $S$. We use $\omega(G)$ to
denote the size of the maximum clique of $G$. A proper coloring of $G$ with
$\chi$ colors is a function $c:V\to \{1,\ldots,\chi\}$ such that for all
$i\in\{1,\ldots,\chi\}$ the graph $G[c^{-1}(i)]$ is an independent set.  We
will focus on the following generalization of coloring:

\begin{definition}

If $\C,\D$ are positive integers then a $(\C,\D)$-coloring of a graph $G=(V,E)$
is a function $c: V \to \{1,\ldots,\C\}$ such that for all $i \in
\{1,\ldots,\C\}$ the maximum degree of $G[c^{-1}(i)]$ is at most $\D$.

\end{definition}

We call the problem of deciding if a graph admits a $(\C,\D)$-coloring, for
given parameters $\C,\D$, \DC. For a graph $G$ and a coloring function $c:
V\to\mathbb{N}$ we say that the \emph{deficiency} of a vertex $u$ is $|N(u)\cap
c^{-1}(c(u))|$, that is, the number of its neighbors with the same color. The
deficiency of a color class $i$ is defined as the maximum deficiency of any
vertex colored with $i$. 

We recall the following basic facts about \DC:

\begin{lemma}\label{lem:trivial}

(\cite{Kang08}) For any $\C,\D$ and any graph $G=(V,E)$ with $\C\cdot\D \ge
|V|$ we have that $G$ admits a $(\C,\D)$-coloring.

\end{lemma}

\begin{proof} Partition $V$ arbitrarily into $\C$ sets of size at most $\left\lceil
\nicefrac{|V|}{\C} \right\rceil$ and color each set with a different color. The maximum
deficiency of any vertex is at most $\left\lceil \frac{|V|}{\C} \right\rceil - 1\le
\frac{|V|}{\C} \le \D$.  \end{proof}

\begin{lemma} \label{lem:basic2}

(\cite{Kang08}) If $G$ admits a $(\C,\D)$-coloring then $\omega(G)\le \C\cdot
(\D+1)$.

\end{lemma}

\begin{proof}
For the sake of contradiction, assume that $G$ has a clique of size
$\C\cdot(\D+1) +1$, then any coloring of $G$ with $\C$ colors must give the
same color to strictly more than $\D+1$ vertices of the clique, which implies
that these vertices have deficiency at least $\D+1$.
\end{proof}

%Graph Classes

Let us now also give some quick reminders regarding the definitions of the
graph classes we consider in this paper.

A graph $G=(V,E)$ is a \emph{split} graph if $V = K\cup S$ where $K$ induces a
clique and $S$ induces an independent set. A graph $G$ is \emph{chordal} if it
does not contain any induced cycles of length four or more. It is well known
that all split graphs are chordal; furthermore it is known that the class of
chordal graphs contains the class of interval graphs, and that chordal graphs
are perfect. For more information on these standard containments see
\cite{BLS99}.
 
A graph is a \emph{cograph} if it is either a single vertex, or the disjoint
union of two cographs, or the complete join of two cographs
\cite{seinsche1974}. Note that the class of cographs is easily seen to be
closed under complement. As a result, complete $k$-partite graphs are cographs
(their complement is a union of cliques, which are themselves cographs); this
fact will be used later. A graph is \emph{trivially perfect} if in all induced
subgraphs the maximum independent set is equal to the number of maximum cliques
\cite{Gol78}.  Trivially perfect graphs are exactly the cographs which are
chordal \cite{YanCC96}, and hence are a subclass of cographs, which are a
subclass of perfect graphs.  We recall that \textsc{Graph Coloring} is
polynomial-time solvable in all the mentioned graph classes, since it is
polynomial-time solvable on perfect graphs \cite{GLS88}, though of course for
all these classes simpler and more efficient combinatorial algorithms are
known.

We will also use the notion of treewidth for the definition of which we refer
the reader to \cite{BodlaenderK08,CyganFKLMPPS15}.

\section{NP-hardness on Cographs} \label{sec:cographs}

In this section we establish that \DC\ is already NP-hard on the very
restricted class of cographs. To this end, we show a reduction from \Part.

\begin{definition} In \Part\ we are given a set $A$ of $4n$ elements, a size
function $s:A\to \mathbb{N}^+$ which assigns a value to each element, and a
target integer $B$.  We are asked if there exists a partition of $A$ into $n$
sets of four elements (quadruples), such that for each set the sum of its
elements is exactly $B$. \end{definition}

\Part\ has long been known to be strongly NP-hard, that is, NP-hard even if all
values are polynomially bounded in $n$. In fact, the reduction given in
\cite{GJ79} establishes the following, slightly stronger statement.

\begin{theorem}\label{thm:partition} \Part\ is strongly NP-complete if $A$ is
given to us partitioned into four sets of equal size $A_1,A_2,A_3,A_4$ and any
valid solution is required to place exactly one element from each $A_i,
i\in\{1,\ldots,4\}$ in each quadruple.  \end{theorem}

\begin{theorem}\label{thm:cograph-hard} \DC\ is NP-complete even when restricted to complete
$k$-partite graphs. Therefore, \DC\ is NP-complete on cographs.  \end{theorem}

\begin{proof}
We start our reduction from an instance of \Part\ where the set of elements $A$
is partitioned into four equal-size sets as in Theorem \ref{thm:partition}. We
first transform the instance by altering the sizes of all elements as follows:
for each element $x\in A_i$ we set $s'(x) := s(x) + 5^i B + 5^5n^2B$ and we
also set $B' := B+B\cdot\sum_{i=1}^4 5^i+ 4\cdot 5^5 n^2B$.  After this
transformation our instance is ``ordered'', in the sense that all elements of
$A_{i+1}$ have strictly larger size than all elements of $A_i$.  Furthermore,
it is not hard to see that the answer to the problem did not change, as any
quadruple that used one element from each $A_i$ and summed up to $B$ now sums
up to $B'$. In addition, we observe that in the new instance the condition that
exactly one element must be used from each set is imposed by the element sizes
themselves: a quadruple that contains two or more elements of $A_4$ will have
sum strictly more than $B'$, while one containing no elements of $A_4$ will
have sum strictly less than $B'$. Similar reasoning can then be applied to
$A_3,A_2$. We note that the element sizes now have the extra property that
$s'(x) \in (B'/4 - 5B'/n^2, B'/4 + 5B'/n^2)$. Indeed, the largest possible size
is at most $M = B+5^4B +5^5Bn^2$, while $B'/4 \ge 5^5Bn^2$, so $M-B'/4 \le 5^4B
+ B \le 5B'/n^2$.  Similarly, the smallest value is at least $L = 1 + 5B +
5^5Bn^2$ while $B'/4 \le 5^5Bn^2 + 5^5B$ so $B'/4-L \le 5^5B$

We now construct an instance of \DC\ as follows. We set $\D=B'$ and $\C=n$.  To
construct the graph $G$, for each element $x\in A_2\cup A_3\cup A_4$ we create
an independent set of $s'(x)$ new vertices which we will call $V_x$. For each
element $x\in A_1$ we construct two independent sets of $s'(x)$ new vertices
each, which we will call $V_x^1$ and $V_x^2$. Finally, we turn the graph into a
complete $5n$-partite graph, that is, we add all possible edges while retaining
the property that all sets $V_x$ and $V_x^i$ remain independent.

Let us now argue for the correctness of the reduction. First, suppose that
there exists a solution to our (modified) \Part\ instance where each quadruple
sums to $B'$. Number the quadruples arbitrarily from $1$ to $n$ and consider
the $i$-th quadruple $(x_i^1, x_i^2, x_i^3, x_i^4)$ where we assume that for
each $j\in\{1,\ldots,4\}$ we have $x_i^j\in A_j$. Hence, $s'(x_i^1)$ is minimum
among the sizes of the elements of the quadruple. We now use color $i$ for all
the vertices of the sets $V_{x_i^j}$ for $j\in\{2,3,4\}$ as well as the sets
$V_{x_i^1}^1, V_{x_i^1}^2$. We continue in this way using a different color for
each quadruple and thus color the whole graph (since the quadruples use all the
elements of $A$). We observe that for any color $i$ the vertices with maximum
deficiency are those from $V_{x_i^1}^1$ and $V_{x_i^1}^2$, and all these
vertices have deficiency exactly $x_i^1+x_i^2+x_i^3+x_i^4 = B'$. Hence, this is
a valid solution.

For the converse direction of the reduction, suppose we are given a
$(\C,\D)$-coloring of the graph we constructed. We first need to argue that
such a coloring must have a very special structure.  In particular, we will
claim that in such a coloring each independent set $V_x$ or $V_x^i$ must be
monochromatic.  Towards this end we formulate a number of claims.

\begin{claim}  

Every color $i$ is used on at most $5B'/4 + 25 B'/n^2 $ vertices. 

\end{claim}

\begin{proof}
We will assume that $i$ is used at least $5B'/4+25B'/n^2+1$ times and obtain a
contradiction. Since the size of the largest independent set $V_x$ is at most
$B'/4 + 5B'/n^2$ we know that color $i$ must appear in at least six different
independent sets.  Among the independent sets in which $i$ appears let $V_x$ be
the one in which it appears the minimum number of times.  The deficiency of a
vertex colored with $i$ in this set is at least $\frac{5}{6} |c^{-1}(i)| \ge
\frac{25B'}{24}>B' = \D$.
\end{proof}

Because of the previous claim, which states that no color appears too many
times, we can also conclude that no color appears too few times. 

\begin{claim}
Every color $i$ is used on at least $5B'/4 - 50 B'/n $ vertices.
\end{claim}

\begin{proof}
First, note that $|V| \ge 5nB'/4 - 25B'/n$ because we have created $5n$
independent sets each of which has size more than $B'/4-5B'/n^2$.  By the
previous claim any color $j\neq i$ has $|c^{-1}(j)| \le 5B'/4 + 25B'/n^2$.
Therefore $\sum_{j\neq i} |c^{-1}(j)| \le (n-1) (5B'/4 + 25B'/n^2)$. We have
$|c^{-1}(i)| = |V| - \sum_{j\neq i} |c^{-1}(j)| \ge \frac{5nB'}{4} -
\frac{25B'}{n} - (n-1)\frac{5B'}{4} - (n-1)\frac{25B'}{n^2} = \frac{5B'}{4} -
\frac{50B'}{n} + \frac{25B'}{n^2} > \frac{5B'}{4} - \frac{50B'}{n}$. 
\end{proof}

Given the above bounds on the size of each color class we can now conclude that
each color appears in exactly five independent sets $V_x$.

\begin{claim}

For each color $i$ the graph induced by $c^{-1}(i)$ is complete $5$-partite.

\end{claim}

\begin{proof}
First, observe that by the previous claim, there must exist at least $5$ sets
$V_x$ or $V_x^i$ that intersect $c^{-1}(i)$, because $|c^{-1}(i)| \ge 5B'/4
-O(B'/n)$ while the size of each such set is at most $B'/4 + O(B'/n^2)$;
therefore, the size of any four sets is strictly smaller than $|c^{-1}(i)|$
(assuming of course that $n$ is sufficiently large). Suppose now that
$c^{-1}(i)$ intersects $6$ different sets, and consider the independent set
$V_x$ on which color $i$ appears at least once but a minimum number of times. A
vertex colored $i$ in this set will have deficiency at least
$\frac{5}{6}(\frac{5B'}{4}-\frac{50B'}{n}) = \frac{25 B'}{24} -
O(\frac{B'}{n})$, which is strictly greater than $B'$ for sufficiently large
$n$. Hence, color $i$ appears in exactly $5$ independent sets. 
\end{proof}

\begin{claim}

In any valid solution every maximal independent set of $G$ is monochromatic.

\end{claim}

\begin{proof}
Consider color $i$, which by the previous claim appears in exactly $5$
independent sets. Suppose that one of these is not monochromatic, say colors
$i,j$ appear in $V_x$. Without loss of generality let $i$ be the minority
color, that is, $i$ appears in at most $|V_x|/2$ vertices of $V_x$. Then we
obtain a contradiction as follows: the total number of times $i$ is used in the
graph is at most $|c^{-1}(i)|\le4(\frac{B'}{4}+\frac{5B'}{n^2}) +
\frac{1}{2}(\frac{B'}{4}+\frac{5B'}{n^2})$, where the first term uses the
general upper bound on the size of all other independent sets where $i$
appears, and the second term uses the same upper bound on $|V_x|$. Thus,
$|c^{-1}(i)| \le \frac{9B'}{8} + O(\frac{B'}{n^2})$ which is strictly smaller
than $\frac{5B'}{4} - \frac{50B'}{n}$, the minimum number of times that $i$
must be used (for sufficiently large $n$).
\end{proof}

We are now ready to complete the converse direction of the reduction. Consider
the vertices of $c^{-1}(i)$, for some color $i$. By the previous sequence of
claims we know that they appear in and fully cover $5$ independent sets $V_x$
or $V_x^i$. We claim that for each $j\in\{2,3,4\}$ any color $i$ is used in
exactly one $V_x$ with $x\in A_j$. This can be seen by considering the
deficiency of the vertices of the smallest independent set where $i$ appears.
The deficiency of these vertices is equal to $x_i^1+x_i^2+x_i^3+x_i^4$, which
are the sizes of the four larger independent sets. By the construction of the
modified \Part\ instance, any quadruple that contains two elements of $A_4$
will have sum strictly greater than $B'$. Hence, these elements must be evenly
partitioned among the color classes, and with similar reasoning the same
follows for the elements of $A_3,A_2$.

We thus arrive at a situation where each color $i$ appears in the independent
sets $V_{x_i^4}, V_{x_i^3}, V_{x_i^2}$ as well as two of the ``small''
independent sets. Recall that all ``small'' independent sets were constructed
in two copies of the same size $V_x^1, V_x^2$. We would now like to ensure that
all color classes contain one small independent set of the form $V_{x_i^1}^1$.
If we achieve this the argument will be complete: we construct the quadruple
$(x_i^4,x_i^3,x_i^2,x_i^1)$ from the color class $i$, and the deficiency of the
vertices of the remaining small independent set ensures that the sum of the
elements of the quadruple is at most $B'$. By constructing $n$ such quadruples
we conclude that they all have sum exactly $B'$, since the sum of all elements
of the \Part\ instance is (without loss of generality) exactly $nB'$.

To ensure that each color class contains an independent set $V_x^1$ we first
observe that we can always exchange the colors of independent sets $V_x^1$ and
$V_x^2$, since they are both of equal size (and monochromatic). Construct now
an auxiliary graph with $\C$ vertices, one for each color class and a directed
edge for each $x\in A_1$. Specifically, if for $x\in A_1$ the independent set
$V_x^1$ is colored $i$ and the set $V_x^2$ is colored $j$ we place a directed
edge from $i$ to $j$ (note that this does not rule out the possibility of
self-loops). In the auxiliary graph, each vertex that does not have a self-loop
is incident on two directed edges. We would like all such vertices to end up
having out-degree 1, because then each color class would contain an independent
set of the form $V_x^1$.  The main observation now is that in each weakly
connected component that contains a vertex $u$ with out-degree 0 there must
also exist a vertex $v$ of out-degree 2. Exchanging the colors of $V_x^1$ and
$V_x^2$ corresponds to flipping the direction of an edge in the auxiliary
graph. Hence, we can take a maximal directed path starting at $v$ and flip all
its edges, while maintaining a valid coloring of the original graph.  This
decreases the number of vertices with out-degree 0 and therefore repeating this
process completes the proof. 
\end{proof}

\section{Polynomial Time Algorithm on Trivially Perfect Graphs}

In this section, we complement the NP-completeness proof from
Section~\ref{sec:cographs}  by giving a polynomial time
algorithm for \DC~on the class of trivially perfect graphs, which form a large subclass of cographs. We will heavily
rely on the following equivalent characterization of trivially perfect graphs
given by \cite{Gol78}:

\begin{theorem} A graph is trivially perfect if and only if it is the
comparability graph of a rooted tree.  \end{theorem}

In other words, for every trivially perfect graph $G$, there exists a rooted
tree $T$ such that making every vertex in the tree adjacent to all of its
descendants yields a graph isomorphic to $G$. We refer to $T$ as the
\textit{underlying rooted tree} of $G$. We recall that it is known how to
obtain $T$ from $G$ in polynomial (in fact linear) time \cite{YanCC96}.

We are now ready to describe our algorithm.  The following observation is one
of its basic building blocks.

\begin{lemma} \label{lem:triv1}

Let $G=(V,E)$ be a trivially perfect graph, $T$ its underlying rooted tree, and
$u\in V$ be an ancestor of $v\in V$ in $T$. Then $N[v]\subseteq N[u]$.

\end{lemma}

\begin{proof}
Any vertex $w\in N[v]$ must be either a descendant of $v$, in which case it is
also a descendant of $u$ and $w\in N[u]$, or another ancestor of $v$. However,
because $T$ is a tree, if $w$ is an ancestor of $v$, then $w$ is either an
ancestor or a descendant of $u$.
\end{proof}

\begin{theorem}\label{thm:triv-perf}
\DC~can be solved in polynomial time on trivially perfect graphs.
\end{theorem}

\begin{proof}
Given a trivially perfect graph $G=(V,E)$ with underlying rooted tree
$T=(V,E')$ and two non-negative integers $\C$ and $\D$, we compute a coloring
of $G$ using at most $\C$ colors and with deficiency at most $\D$ as follows.
First, we partition the vertices of $T$ (and therefore of $G$) into sets
$V_1,\ldots,V_\ell$, where $\ell$ denotes the height of $T$, such that $V_1$
contains the leaves of $T$ and, for every $2\leq i \leq \ell$, $V_i$ contains
the leaves of $T\setminus (\bigcup_{j=1}^{i-1} V_j)$. Observe that each set
$V_i$ is an independent set in $G$. We now start our coloring by giving all
vertices of $V_1$ color 1. Then, for every $2\leq i\leq \ell$, we color the
vertices of $V_i$ by giving each of them the lowest color available, i.e., we
color each vertex $u$ with the lowest $j$ such that $u$ has at most $\D$
descendants colored $j$. If for some vertex no color is available, that is, its
subtree contains at least $\D+1$ vertices colored with each of the colors
$\{1,\ldots,\C\}$, then we return that $G$ does not admit a $(\C,\D)$-coloring.

This procedure can clearly be performed in polynomial time and, if it returns a
solution, it uses at most $\C$ colors. Furthermore, whenever the procedure uses
color $i$ on a vertex $u$ it is guaranteed that $u$ has deficiency at most $\D$
among currently colored vertices. Since any neighbor of $u$ that is currently
colored with $i$ must be a descendant of $u$, by Lemma \ref{lem:triv1} this
guarantees that the deficiency of all vertices remains at most $\D$ at all
times. 

It now only remains to prove that the algorithm concludes that $G$ cannot be
colored with $\C$ colors and deficiency $\D$ only when no such coloring exists.
For this we will rely on the following claim which states that any valid
coloring can be ``sorted''.

\begin{claim} If $G$ admits a $(\C,\D)$-coloring, then there exists a
$(\C,\D)$-coloring of $G$ $c$ such that, for every two vertices $u,v\in V(G)$,
if $v$ is a descendant of $u$, then $c(v)\le c(u)$.  \end{claim}

\begin{proof}
Let us consider an arbitrary $(\C,\D)$-coloring $c^*: V(G) \rightarrow
\{1,\ldots,\C\}$ of $G$. We describe a process which, as long as there exist
$u,v\in V$ with $u$ an ancestor of $v$ and $c^*(u) < c^*(v)$ transforms $c^*$ to
another valid coloring which is closer to having the desired property. So,
suppose that such a pair $u,v$ exists, and furthermore, if many such pairs
exist, suppose that we select a pair where $u$ is as close to the root of $T$
as possible. As a result, we can assume that no ancestor $u'$ of $u$ has color
$c^*(u)$, because otherwise we would have started with the pair $u',v$.

We will now consider two cases.  Assume first that there exists a vertex $x$
such that $c^*(x)=c^*(v)$ and $x$ is an ancestor of $u$. We claim that swapping
the colors of $u$ and $v$ yields a new coloring of $G$ with deficiency at most
$\D$. The only affected vertices are those colored $c^*(u)$ or $c^*(v)$.
Regarding color $c^*(u)$, because by Lemma \ref{lem:triv1} $N[v]\subseteq N[u]$
and color $c^*(u)$ was moved from $u$ to $v$, the deficiency of every vertex
colored $c^*(u)$ in $V\setminus\{u,v\}$ is at most as high as it was before,
and the deficiency of $v$ is at most as high as the deficiency of $u$ in $c^*$.
Regarding color $c^*(v)$ we observe that the deficiency of vertex $x$ remains
unchanged, since both $u,v$ are its neighbors, and the same is true for all
ancestors of $x$. Since the deficiency of $x$ is at most $\D$, by Lemma
\ref{lem:triv1}, the deficiency of every descendant of $x$ colored with
$c^*(v)$ is also at most $\D$.

For the remaining case, suppose that no ancestor of $u$ uses color $c^*(v)$.
Recall that we have also assumed that no ancestor of $u$ uses color $c^*(u)$.
We therefore transform the coloring as follows: in the subtree rooted at $u$ we
exchange colors $c^*(u)$ and $c^*(v)$ (that is, we color all vertices currently
colored with $c^*(u)$ with $c^*(v)$ and vice-versa). Because no ancestor of $u$
uses either of these two colors, this exchange does not affect the deficiency
of any vertex.

We can now repeat this procedure as follows: as long as there is a conflicting
pair $u,v$, with $u$ an ancestor of $v$ and $c^*(u)<c^*(v)$ we select such a
pair with $u$ as close to the root as possible and, if there are several such
pairs, we select the one with maximum $c^*(v)$. We perform the transformation
explained above on this pair and then repeat. It is not hard to see that every
vertex will be used at most once as an ancestor in this transformation,
because after the transformation it will have the highest color in its subtree.
Hence we will eventually obtain the claimed property.  
\end{proof}

It follows from the previous claim that if a $(\C,\D)$-coloring exists, then a
sorted $(\C,\D)$-coloring exists where ancestors always have colors at
least as high as their descendants. We can now argue that our algorithm also
produces a sorted coloring, with the extra property that whenever it sets
$c(u)=i$ we know that \emph{any} sorted $(\C,\D)$-coloring of $G$ must give
color at least $i$ to $u$. This can be shown by induction on $i$: it is clear
for the vertices of $V_1$ to which the algorithm gives color $1$; and if the
algorithm assigns color $i$ to $u$, then $u$ has $\D+1$ descendants which (by
inductive hypothesis) must have color at least $i-1$ in any valid sorted
coloring of $G$.
\end{proof}

\section{Algorithms on Cographs}

In this section we present algorithms that can solve \DC\ on cographs in
polynomial time if either $\D$ or $\C$ is bounded; both algorithms rely on
dynamic programming.  After presenting them we show how their combination
can be used to obtain a sub-exponential time algorithm for \DC\ on cographs.

\subsection{Algorithm for Small Deficiency} \label{sec:DP1}

We now present an algorithm that solves \DC\ in polynomial time on cographs if
$\D$ is bounded. Before we proceed, let us sketch the main ideas behind the
algorithm.  Given a $(\C,\D)$-coloring $c$ of a graph $G$, we define the
\emph{type} of a color class $i$, as the pair of two integers $(s_i,d_i)$ where
$s_i:=\min\{|c^{-1}(i)|, \D+1\}$ and $d_i$ is the maximum degree of
$G[c^{-1}(i)]$.  In other words, the type of a color class is characterized by
its size (up to value $\D+1$) and the maximum deficiency of any of its
vertices. We observe that there are at most $(\D+1)^2$ possible types in a
valid $(\C,\D)$-coloring, because $s_i$ only takes values in
$\{1,\ldots,\D+1\}$ and $d_i$ in $\{0,\ldots,\D\}$.

We can now define the \emph{signature} of a coloring $c$ as a tuple which
contains one element for every possible color type $(s,d)$. This element is the
number of color classes in $c$ that have type $(s,d)$, and hence is a number in
$\{0,\ldots,\C\}$. We can conclude that there are at most $(\C+1)^{(\D+1)^2}$
possible signatures that a valid $(\C,\D)$-coloring can have. Our algorithm
will work via dynamic programming, using the fact that any cograph can be seen
as a union or join of two of its induced subgraphs. We therefore analyze a
given cograph into its constituent subgraphs in this way (obtaining a binary
tree) and for each such graph we will construct a binary table which states for
each possible signature if the current graph admits a $(\C,\D)$-coloring with
this signature. The obstacle now is to describe a procedure which, given two
such tables for graphs $G_1,G_2$ is able to generate the table of admissible
signatures for their union and their join. If we do this we can inductively
compute such a table for all intermediate graphs used in the construction of
the input graph $G$ and eventually for $G$ itself.

\begin{theorem}\label{thm:DP1} There is an algorithm which, given a cograph $G$ and integers $\C,\D$, decides if $G$ admits a
$(\C,\D)$-coloring in time $O^*\left(\C^{O((\D)^4)}\right)$. \end{theorem}

\begin{proof}
We use the ideas sketched above. Specifically, we say that a coloring signature
$S$ is a function $\{1,\ldots,\D+1\}\times \{0,\ldots,\D\} \to \{0,\ldots,\C\}$
and a coloring $c$ has signature $S$ if for any $(s,d)\in
\{1,\ldots,\D+1\}\times \{0,\ldots,\D\}$ the number of color classes with type
$(s,d)$ in $c$ is $S((s,d))$. Our algorithm will maintain a binary table $T$
with the property that, for $S$ a possible coloring signature we have $T(S)=1$
if and only if there exists a $(\C,\D)$-coloring of $G$ with signature $S$. The
size of $T$ is therefore at most $(\C+1)^{(\D+1)^2}$.

It is not hard to see how to compute $T$ if $G$ consists of a single vertex:
the only color class then has type $(1,0)$, so the only possible signature is
the one that sets $S((1,0))=1$ and $S((s,d))=0$ otherwise.

Now, suppose that $G$ is either the union or the join of two graphs $G_1,G_2$
for which our algorithm has already calculated the corresponding tables
$T_1,T_2$. We will use the fact that for any valid $(\C,\D)$-coloring $c$ of
$G$ with signature $S$, its restrictions to $G_1,G_2$ are also valid
$(\C,\D)$-colorings. If these restrictions have signatures $S_1,S_2$ it must
then be the case that $T_1(S_1)=T_2(S_2)=1$. It follows that in order to
compute all the signatures for which we must have $T(S)=1$ it suffices to
consider all pairs of signatures $S_1,S_2$ such that $T_1(S_1)=T_2(S_2)=1$ and
decide if it is possible to have a coloring of $G$ with signature $S$ whose
restrictions to $G_1,G_2$ have signatures $S_1,S_2$.

Given two signatures $S_1,S_2$ such that $T_1(S_1)=T_2(S_2)=1$ we would
therefore like to generate all possible signatures $S$ for colorings $c$ of $G$
such that $S_1,S_2$ represent the restriction of $c$ to $G_1,G_2$. Every color
class of $c$ will either consist of vertices of only one subgraph $G_1$ or
$G_2$, or it will be the result of merging a color class of $G_1$ with a color
class of $G_2$. Our algorithm will enumerate all possible merging combinations
between color classes of $G_1$ and $G_2$. 

Let us now explain how we enumerate all merging possibilities. Let $c_1,c_2$ be
$(\C,\D)$-colorings of $G_1,G_2$ with signatures $S_1,S_2$ respectively. In the
remainder of this proof we explain how given these two signatures we can
generate all possible signatures of feasible colorings $c$ of the union or join
of $G_1,G_2$, such that $c$ has signatures $S_1,S_2$ when restricted to
$G_1,G_2$ respectively.  Essentially, the problem is that some color classes of
$c_1$ can be merged with some color classes of $c_2$ to produce color classes
of the new graph, so we need to consider all possible combinations in which
this can happen.

If $G$ is the join of $G_1,G_2$ we say that type $(s_1,d_1)$ is mergeable with
type $(s_2,d_2)$ if $s_1+d_2\le \D$ and $s_2+d_1\le \D$. If $G$ is the union of
$G_1,G_2$ we say that any pair of types is mergeable. Furthermore, if $G$ is
the join of $G_1,G_2$ and $(s_1,d_1), (s_2,d_2)$ are mergeable types, we say
that they merge into type $(\min\{s_1+s_2, \D+1\}, \max\{d_1+s_2, d_2+s_1\})$.
If $G$ is the union of $G_1,G_2$ we say that types $(s_1,d_1)$ and $(s_2,d_2)$
merge into type $(\min\{s_1+s_2, \D+1\}, \max\{d_1,d_2\})$.  The intuition
behind these definitions is that a color class $i$ of $c_1$ is mergeable with a
color class $j$ of $c_2$ if we can use a single color for $c_1^{-1}(i)\cup
c_2^{-1}(j)$ in $G$, and the type of this color class is the type into which
the types of $i,j$ merge.

Now, in order to enumerate all merging possibilities we construct an auxiliray
bipartite graph $G'=(A_1,A_2,E')$.  The graph $G'$ consists of $(\D+1)^2$
vertices on each side, each corresponding to a type. We place an edge between
two vertices if their corresponding types are mergeable (so if $G$ is a union
of $G_1,G_2$ then $G'$ is a complete bipartite graph). We also give a weight to
each vertex as follows: if $u\in A_i$ corresponds to type $(s,d)$ we set
$w(u)=S_i((s,d))$. In words, the weight of a vertex that represents a type is
the number of color classes of that type in the coloring of the subgraphs.

We will now enumerate assignments of non-negative weights to the edges of $G'$
which satisfy the condition that for all vertices $u\in A_1\cup A_2$ we have
$\sum_{v\in N(u)} w( (u,v) ) \le w(u)$. The idea here is that if we increase
the weight of the edge $(u,v)$ by one, we mean that we merge a color that has
type $u$ in $c_1$ with a color that has type $v$ in $c_2$ to produce a color
class in $G$. The constraint we imposed therefore means that the total number
of times we do this for color classes of type $u$ cannot be higher than $w(u)$,
which is the number of colors that have this type.  It is not hard to see that
the total number of valid edge-weight assignments is at most
$(\C+1)^{(\D+1)^4}$, since every edge must receive weight in $\{0,\ldots,\C\}$
and there are at most $(\D+1)^4$ edges. This is the step that dominates the
running time of our algorithm.

For each of the enumerated assignments of $G'$ we can now calculate a signature
$S$ of a coloring of $G$. For each type $(s,d)$ let $E_{(s,d)}\subseteq E'$ be
the set of edges of $G'$ whose endpoints merge into type $(s,d)$. Let $u_i\in
A_i$ be the vertices corresponding to type $(s,d)$. We have $S((s,d)) =
\sum_{i=1,2} (w(u_i) - \sum_{v\in N(u_i)} w(u_i,v)) + \sum_{e\in E_{(s,d)}}
w(e)$. We now check that the signature we computed refers to a coloring with at
most $\C$ colors, that is, if $\sum S((s,d)) \le \C$, where
$s\in\{1,\ldots,\D+1\}$ and $d\in\{0,\ldots,\D\}$.  In this case we set
$T(S)=1$.  The observation that completes the proof is that for all valid
colorings $c$ of $G$ with signature $S$ such that the restriction of $c$ to
$G_1,G_2$ has signatures $S_1,S_2$ there must exist a weight assignment for
which the above procedure finds the signature $S$. Hence, by examining all
pairs of feasible signatures $S_1,S_2$ we will discover all feasible signatures
of $G$.   \end{proof}

\subsection{Algorithm for Few Colors}\label{sec:DP2}

In this section we provide an algorithm that solves \DC\ in polynomial time on
cographs if $\C$ is bounded. The type of a color class $i$ is defined in a
similar manner as in the first paragraph of Section \ref{sec:DP1}.
Specifically, for a given coloring $c$, we define the type of color $i$ as the
pair of two integers $(s_i,d_i)$ where $s_i:=\min\{|c^{-1}(i)|, \D+1\}$ and
$d_i$ is the maximum degree of $G[c^{-1}(i)]$. Note that $s_i\in \{0, \ldots,
\D+1\}$. What changes in this section is the signature $S$ of a coloring $c$
which is now defined as a function $S:\{1,\ldots,\C\} \rightarrow \{0, \ldots,
\D+1\} \times \{0, \ldots, \D\}$, which takes as input a color class and
returns its type.  

Once again, we will construct a dynamic programming a table $T$ which given a
signature $S$ set $T(S)=1$ if the current graph has a $(\C,\D)$-coloring with
signature $S$. Our table will have size at most $(\D + 2)^{2\C}$, since this is
an upper bound on the number of possible signatures.  As in the previous
section, we shall describe how to compute table $T$ of a graph $G$ which is the
union or the join of two graphs $G_1$ and $G_2$ whose tables $T_1$ and $T_2$
are known.

\begin{theorem}\label{thm:DP2} There is an algorithm which, given a cograph $G$ and integers $\C,\D$, decides if $G$ admits a
$(\C,\D)$-coloring in time $O^*\left( (\D)^{O(\C)}\right)$. 
\end{theorem}

\begin{proof}
The base case is when we have a single vertex $u$ in the graph $G$. In this case, any coloring of $u$ is valid, so for all $i\in\{1,\ldots, \C \}$ we define a signature $S_i$ such that $S_i(i) = (1,0)$ and $S_i(j) = (0,0)$ when $i\neq j$. Last, $T(S) = 1$ if and only if $S = S_i$ for any $i$.

Now, suppose that $G$ is either the union or the join of two graphs $G_1,G_2$
for which we have already calculated their corresponding tables. Once again we
just need to consider all pairs of signatures $S_1,S_2$ such that
$T_1(S_1)=T_2(S_2)=1$ and decide if we can have a coloring of $G$ with
signature $S$ whose restrictions to $G_1,G_2$ have signatures $S_1,S_2$. Let
$S_1, S_2$ be one such pair of signatures, for which $S_j(i) = (s_j^i, d_j^i)$,
$j\in\{1,2\}$. We examine the cases of union and join separately.

Let us start with the case that $G$ is the union of $G_1,G_2$. Define $S$ such that for any $i$, $S(i) = (\min\{s_1^i+s_2^i,\D +1\}, \max\{d_1^i,d_2^i\} )$ and set $T(S)=1$.

The case where $G$ is the join of $G_1, G_2$ is a little more complicated since we first need to check if, given two precolored graphs the outcome of their join is valid, that for all colors $i$, the maximum degree of $G[i]$ remains at most $\D$. This corresponds to checking for all colors $i$ whether $d = \max\{s_1^i+d_2^i,s_2^i+d_1^j\}\le \D$. Given that the above is true, we define $S$ such that for any $i$, $S(i)= (\min\{s_1^i+s_2^i,\D +1\}, d)$ and set $S(T)=1$.

The algorithm considers all pairs of elements of $T_1, T_2$, so it runs in time dominated by $|T_i|^2 = O^*\left( (\D)^{O(\C)}\right)$. 
\end{proof}

\subsection{Sub-Exponential Time Algorithm}\label{sec:subexp}

We now combine the algorithms of Sections \ref{sec:DP1} and \ref{sec:DP2} in order to obtain a sub-exponential time algorithm for cographs.

\begin{theorem}\label{thm:subexp}
\DC\ can be solved in time $n^{O\left(n^{\nicefrac{4}{5}}\right)}$ on cographs.
\end{theorem}

\begin{proof}
First, we remind the reader that, from Lemma \ref{lem:trivial}, if $\D \cdot \C \ge n$ then the answer is trivially yes. Thus the interesting case is when $\D \cdot \C < n$. Note that we also trivially have that $\D,\C \le n$.

If $\D \le \sqrt[5]{n}$, then the algorithm of Section \ref{sec:DP1} runs in $O^*\left(\C^{O((\D)^4)} \right) = n^{O\left(n^{\nicefrac{4}{5}}\right)}$ time.

If $\D > \sqrt[5]{n}$, then $\C < \frac{n}{\D} < n^{\frac{4}{5}}$. In this case, the algorithm of Section \ref{sec:DP2} runs in $O^*\left((\D)^{O(\C)}\right) = n^{O\left(n^{\nicefrac{4}{5}}\right)}$ time. 
\end{proof}

\section{Split and Chordal Graphs}

In this section we present the following results: first, we show that \DC\ is
hard on split graphs even when $\D$ is a fixed constant, as long as $\D\ge 1$;
the problem is of course in P if $\D=0$. Then, we show that \DC\ is hard on
split graphs even when $\C$ is a fixed constant, as long as $\C\ge 2$; the
problem is of course trivial if $\C=1$. These results completely describe the
complexity of the problem when one of the two relevant parameters is fixed. We
then give a treewidth-based procedure through which we obtain a polynomial-time
algorithm even on chordal graphs when $\C,\D$ are bounded (in fact, the
algorithm is FPT parameterized by $\C+\D$). Hence these results give a complete
picture of the complexity of the problem on chordal graphs: the problem is
still hard when one of $\C,\D$ is bounded, but becomes easy if both are
bounded.

Let us also remark that both of the reductions we present are linear. Hence,
under the Exponential Time Hypothesis \cite{IPZ01}, they establish not only
NP-hardness, but also unsolvability in time $2^{o(n)}$ for \DC\ on split
graphs, for constant values of $\C$ or $\D$. This is in contrast with the
results of Section \ref{sec:subexp} on cographs.

\subsection{Hardness for Bounded Deficiency}

In this section we show that \DC\ is NP-hard for any fixed value $\D\ge 1$. We
will describe a reduction from 3CNFSAT which for any given $\D\ge 1$ produces
an instance of \DC\ that has a $(2n,\D)$-coloring if the given formula is
satisfiable.  Suppose we are given a CNF formula $\phi$ where $X= \{x_1,
\ldots, x_n\}$ are the variables and $C = \{c_1, \ldots, c_m\}$ are the clauses
and each clause contains exactly 3 literals. We construct a graph $G$ as
follows: 

\begin{itemize}

\item For each $i\in\{1,\ldots,n\}$ we construct a set $U_i$ of $2\D+2$
vertices; we partition $U_i$ into two sets $U_i^A$, $U_i^D$, each of size $\D$,
and two single vertices $u_i^B, u_i^C$. Let $U=\bigcup_i U_i$.

\item For each $i\in\{1,\ldots,n\}$ we construct a vertex $z_i^A$, a vertex
$z_i^B$, and a vertex $z_i^D$.

\item For each $j\in\{1,\ldots,m\}$ we construct a vertex $v_j$.

\item We turn the vertices of $U$ into a clique.

\item For each $i\in\{1,\ldots,n\}$ we connect $z_i^A$ to all of $U$ except
$U_i^D\cup \{u_i^B,u_i^C\}$; we connect $z_i^B$ to all of $U$ except $U_i^A\cup
U_i^D\cup \{u_i^C\}$; we connect $z_i^D$ to all of $U$ except $U_i^A\cup
\{u_i^B,u_i^C\}$.

\item For each $j\in\{1,\ldots,m\}$ we connect $v_j$ to all of $U$ except for
the following: for each variable $x_i$ that appears in $c_j$ positive, $v_j$ is
\emph{not} connected to $U_i^A\cup \{u_i^B\}$; for each variable $x_i$ that
appears in $c_j$ negative, $v_j$ is \emph{not} connected to $U_i^A\cup
\{u_i^C\}$.

\end{itemize}

Observe that the graph $G$ we have constructed is split as $U$ is a clique and
the remaining vertices are an independent set. We claim that this graph has a
$(2n,\D)$-coloring if and only if $\phi$ is satisfiable.

\begin{lemma}\label{lem:sd2f} If $\phi$ is satisfiable then $G$ has a
$(2n,\D)$-coloring.  \end{lemma}

\begin{proof}
We first color $U$ as follows: fix a satisfying assignment and consider the
variable $x_i$. If $x_i$ is set to True then we color $U_i^A\cup \{u_i^B\}$
with color $2i-1$ and $U_i^D\cup\{u_i^C\}$ with color $2i$; otherwise we color
$U_i^A\cup\{u_i^C\}$ with color $2i-1$ and $U_i^D\cup\{u_i^B\}$ with color
$2i$. Note that now every vertex of the clique $U$ has deficiency exactly $\D$
and we have used all colors.

We assign to vertex $z_i^A$ color $2i$ and to vertex $z_i^D$ color $2i-1$.
Observe that $z_i^A$ is not connected to $U_i^D\cup\{u_i^B, u_i^C\}$, which are
the only vertices which may have color $2i$ at this point, so its deficiency is
$0$ and it does not affect the deficiency of any other vertex. Similarly,
$z_i^D$ has deficiency $0$. We assign to $z_i^B$ the same color as $u_i^C$ and
this vertex also has deficiency $0$.

Finally, for each $j\in\{1,\ldots,m\}$ we consider the $j$-th clause of $\phi$
in order to color $v_j$. Since the assignment is satisfying, this clause
contains a true literal, say involving the variable $x_i$. We assign to $v_j$
the color $2i-1$. If $x_i$ appears positive in clause $c_j$ then $v_j$ is not
connected to $U_i^A\cup \{u_i^B\}$, which are all the vertices that have color
$2i-1$, so its deficiency is $0$. The reasoning is similar if $x_i$ appears
negative in $c_j$.  \end{proof}

\begin{lemma}\label{lem:sd2b} If $G$ has a $(2n,\D)$-coloring, then $\phi$ is
satisfiable.  \end{lemma}

\begin{proof}
First, observe that $U$ is a clique of size $2n(\D+1)$, so the given coloring
must use all colors in $U$, each $\D+1$ times. Furthermore, all vertices of $U$
have deficiency exactly $\D$ already in $U$, so they cannot have any neighbors
with the same color outside of $U$.

Because of the previous arguments, the color used in $z_i^A$ can only appear in
$U_i$ inside the clique (as $z_i^A$ is connected to the rest of $U$) and must
appear there $\D+1$ times. The same argument applies to $z_i^B, z_i^D$.
Therefore, if two distinct colors are used for the vertices in $z_i^A, z_i^B,
z_i^D$, these colors cover all of $U_i$ and are not used anywhere else in $U$.
In this case we say that the coloring of $U_i$ is normal. Furthermore, it is
impossible to use three distinct colors for $z_i^A, z_i^B, z_i^D$.

Suppose that some $U_i$ has a coloring that is not normal. Then $z_i^A, z_i^B,
z_i^D$ must use the same color. This color appears $\D+1$ times in $U$ but
cannot appear in $U\setminus U_i$. Furthermore, it cannot appear in a neighbor
of $z_i^A, z_i^B$ or $z_i^D$. It can therefore not appear in $U_i^A$ (which is
connected to $z_i^A$), nor in $U_i^D$ (which is connected to $z_i^D$), nor in
$u_i^B$ (which is connected to $z_i^B$). So the only vertex where it may appear
is $u_i^C$, contradicting the assumption that this color is used $\D+1$ times
in the clique.  We conclude that all $U_i$ must have a normal coloring.

Since all $U_i$ have a normal coloring, the $2\D+2$ vertices of each $U_i$ are
colored with two colors, each used $\D+1$ times. Suppose without loss of
generality that for all $i$ the set $U_i$ is colored with colors $2i-1$ and
$2i$ and furthermore that $2i-1$ is used at least once in $U_i^A$. We claim
that $U_i^A$ must in fact be monochromatic and colored completely with $2i-1$.
This is trivially true if $\D=1$; while when $\D\ge 2$ if $U_i^A$ uses both
colors $2i-1$ and $2i$, then no color is available for $z_i^A$.  We now obtain
an assignment to $\phi$ as follows: we set $x_i$ to True if $u_i^B$ has color
$2i-1$; otherwise we set $x_i$ to False.

We claim that this assignment satisfies $\phi$. Consider the $j$-th clause and
let $r$ be its assigned color. Suppose $r$ appears in the colors of $U_i$, so
$r=2i-1$ or $r=2i$. It must be the case that $r$ is not used in any of the
neighbors of $v_j$, therefore the variable $x_i$ appears in $c_j$. Suppose
$r=2i$. Then, the color $2i$ cannot appear in $U_i^D$ (which is connected to
$v_j$), which implies that $U_i^A\cup U_i^D$ are colored fully with $2i-1$.
This can only happen if $\D=1$, but in this case $u_i^B, u_i^C$ are both
colored $2i$. Since one of them is connected to $v_j$ we have a contradiction.
Therefore, $v_j$ must be colored with $2i-1$. If $v_j$ is not connected to
$u_i^B$ (so $x_i$ appears positive in $c_j$) then $u_i^C, u_i^D$ have color
$2i$, which means that $u_i^B$ has color $2i-1$ and our assignment satisfies
the clause. Similarly, if $v_j$ is not connected to $u_i^C$ (so $x_i$ appears
negative in $c_j$), then $u_i^B$ must have color $2i$, so our assignment sets
$x_i$ to False and satisfies $c_j$. 
 \end{proof}

The main theorem of this section follows from Lemmas \ref{lem:sd2f} and
\ref{lem:sd2b}.

\begin{theorem}\label{thm:sd2} \DC\ is NP-hard on split graphs for any fixed
$\D\ge 1$.  \end{theorem}

\subsection{Hardness for Bounded Number of Colors}

\begin{theorem}\label{thm:split-colors}
\DC\ is NP-complete on split graphs for every fixed value of $\C \geq 2$.
\end{theorem}

\begin{proof}
We reduce from the problem \textsc{3-Set Splitting}, which takes as input a set
of elements $U$, called the \textit{universe}, and a family $\cal{F}$ of
subsets of $U$ of size exactly~3, and asks whether there is a partition
$(U_1,U_2)$ of $U$ such that, for every set $S\in \cal{F}$, we have $S \cap U_1
\neq \emptyset$ and $S \cap U_2 \neq \emptyset$.  This problem is well-known to
be NP-complete~\cite{GJ79}.  Given an instance $(U,\cal{F})$ of \textsc{3-Set
Splitting} and a positive integer $\C$, we build a split graph $G=(V,E)$ such
that $V = C_1 \cup C_2 \cup C^* \cup I \cup Z_1 \cup Z_2$, with $|C_1|=|C_2| =
|\mathcal{F}|$, $|C^*| = (\C-2) \cdot (|\mathcal{F}|+2)$, $|I|=|U|$ and $|Z_1|=
|Z_2| = \C \cdot (|\mathcal{F}|+2)$.  We proceed by making all the vertices of
$C_1 \cup C_2$ pairwise adjacent. We then associate each set $S$ of $\cal{F}$
with two vertices $v_S^1$ and $v_S^2$ of $C_1$ and $C_2$ respectively, and
every element $x$ of $U$ with a vertex $w_x$ of $I$.  For every pair $x\in U,
S\in \cal{F}$, we make $w_x$ adjacent to $v_S^1$ and $v_S^2$ if and only if
$x\in S$.  We complete our construction by making all the vertices of $C_i$
adjacent to all the vertices of $Z_i$ for $i\in \{1,2\}$, and all the vertices
of $C^*$ adjacent to every vertex in $V$. Observe that the graph we constructed
is split, since $C_1\cup C_2\cup C^*$ induces a clique and $I\cup Z_1\cup Z_2$
induces an independent set.  We now claim that there exists a partition
$(U_1,U_2)$ of $U$ as described above if and only if $G$ can be colored with at
most $\C$ colors and deficiency at most $\D = |\mathcal{F}|+1$.\\

For the forward direction, we color every vertex of $C_1 \cup Z_2$ with color 1
and every vertex of $C_2 \cup Z_1$ with color 2. For each $w_x\in I$ we color
$w_x$ with $1$ if $x\in U_1$ and with $2$ if $x\in U_2$. We color the vertices
of $C^*$ equitably with the remaining $\C-2$ colors. This produces the desired
coloring of $G$.\\

For the converse, we first prove the following:

\begin{claim}  
For any coloring of $G$ with $\C$ colors and deficiency at most $\D$, all the vertices of $C_1$ have the same color. Similarly, all the vertices of $C_2$ have the same color, and this color is distinct from that of the vertices of $C_1$. Additionally, the remaining $\C - 2$ colors are each used exactly $\D + 1$ times in $C^*$.
\end{claim}

\begin{proof}
We first consider the colors given to the vertices of $Z_1$ and $Z_2$. Observe that since both sets have size $\C \cdot (|\mathcal{F}|+2) = \C \cdot (\D + 1)$, there is a color $c_1$ that appears at least $\D + 1$ times in $Z_1$ and a color $c_2$ that appears at least $\D + 1$ times in $Z_2$. Since $Z_1 \subset N(u)$ for every vertex $u\in C_1 \cup C^*$, we obtain that no vertex of $C_1 \cup C^*$ uses color $c_1$. Using a similar argument, we obtain that no vertex of $C_2 \cup C^*$ uses color $c_2$. 

We will first prove that $c_1\neq c_2$. Indeed, suppose that $c_1 = c_2$. Since
this color $c_1$ does not appear in $C_1\cup C_2 \cup C^*$, we are left with
$\C-1$ available colors for these sets, where $|C_1\cup C_2 \cup C^*|$ = $\C
\cdot |\mathcal{F}| +2\C -4$. To obtain a contradiction observe that at least
one color class should have size at least $\frac{|C_1\cup C_2 \cup C^*|}{\C -
1} > |\mathcal{F}| + 2$ for sufficiently large $\mathcal{F}$, which is more
than $\D +1$ vertices. 

The above implies that $C^*$ must be colored using at most $\C - 2$ colors. Since $C^*$ is a clique of size exactly $(\C-2) \cdot (|\mathcal{F}|+2) = (\C-2) \cdot (\D+1)$, it follows that $C^*$ is colored using $\C - 2$ colors, each of which are used exactly $\D + 1$ times, as desired. Last, we conclude that vertices in $C_1$ should only be colored $c_2$ and similarly vertices of $C_2$ should receive color $c_1$. 
\end{proof}

By the previous claim $C_1$ and $C_2$ are both monochromatic and use different colors. Without loss of generality suppose that $C_1$ is colored with color $1$ and $C_2$ with color $2$. Since every vertex of $I$ is adjacent to every vertex of $C^*$, we immediately obtain that $I$ must be colored using only $1$ or $2$. It only remains to show that for every set $S$ of $\cal{F}$, there exist elements $x,y\in S$ such that vertex $w_x$ is colored with color $1$ and vertex $w_y$ is colored with color $2$. Then, the coloring of $I$ will give us a partition of $U$. Assume for contradiction that there exists a set $S$ whose elements $x,y$ and $z$ all have the same color, say color $1$. From the above claim, we know that $v_S^1\in C_1$ uses color $1$ and is adjacent to the other $\D - 2$ vertices of $C_1$, all of which also use color $1$. Therefore, $v_S^1$ is adjacent to $\D - 2 + 3$ vertices using color $1$, and hence has deficiency $\D+1$, a contradiction. This concludes the proof.
\end{proof}

\subsection{A Dynamic Programming Algorithm}

In this section we present an algorithm which solves the problem efficiently on
chordal graphs when $\C$ and $\D$ are small. Our main tool is a treewidth-based
procedure, as well as known connections between the maximum clique size and
treewidth of chordal graphs.

\begin{theorem} \label{thm:tw}

\DC\ can be solved in time $(\C\D)^{O(\TW)}n^{O(1)}$ on any graph $G$ with $n$
vertices if a tree decomposition of width $\TW$ of $G$ is supplied with the
input.

\end{theorem}

\begin{proof}
We describe a dynamic programming algorithm which uses standard techniques, and
hence we sketch some of the details. Suppose that we are given a rooted nice
tree decomposition of $G$ (we use here the definition of nice tree
decomposition given in \cite{BodlaenderK08}).  For every bag $B$ of the
decomposition we denote by $B^\downarrow$ the set of vertices of $G$ that
appear in $B$ and bags below it in the decomposition. For a coloring
$c:V\to\{1,\ldots,\C\}$ we say that the partial type of a vertex $u\in B$ is a
pair consisting of $c(u)$ and $|c^{-1}(c(u))\cap N(u) \cap B^\downarrow|$. In
words, the type of a vertex is its color and its deficiency in the graph
induced by $B^\downarrow$.  Clearly, if $c$ is a valid coloring, any vertex can
have at most $\C\cdot(\D+1)$ types.  Hence, if we define the type of $B$ as a
tuple containing the types of its vertices, any bag can have one of at most
$(\C\cdot(\D+1))^{\TW}$ types.

Our dynamic programming algorithm will now construct a table which for every
bag $B$ and every possible bag type decides if there is a coloring of
$B^\downarrow$ with the specified type for which all vertices of
$B^\downarrow\setminus B$ have deficiency at most $\D$. The table is easy to
construct for leaf bags and forget bags. For introduce bags we consider all
possible colors of the new vertex, and for each color we appropriately compute
its deficiency and update the deficiency of its neighbors in the bag, rejecting
solutions where a vertex reaches deficiency $\D+1$. Finally, for join bags we
consider any pair of partial solutions from the two children bags that agree on
the colors of all vertices of the bag and compute the deficiency of each vertex
as the sum of its deficiencies in the two solutions.
\end{proof}

We now recall the following theorem connecting $\omega(G)$ and $\TW(G)$ for
chordal graphs.

\begin{theorem}\label{thm:chordal-basic}

(\cite{RS86,Bodlaender98}) In chordal graphs $\omega(G)=\TW(G)+1$. Furthermore, an
optimal tree decomposition of a chordal graph can be computed in polynomial
time.

\end{theorem}

Together with Lemma \ref{lem:basic2} this gives the following algorithm for
chordal graphs.

\begin{theorem}\label{thm:split-alg}

\DC\ can be solved in time $(\C\D)^{O(\C\D)}n^{O(1)}$ in chordal graphs.

\end{theorem}

\begin{proof}
We use Theorem \ref{thm:chordal-basic} to compute an optimal tree decomposition
of the input graph and its maximum clique size. If $\omega(G)>\C(\D+1)$ then we
can immediately reject by Lemma \ref{lem:basic2}. Otherwise, we know that
$\TW(G)\le \C(\D+1)$ from Theorem \ref{thm:chordal-basic}, so we apply the
algorithm of Theorem \ref{thm:tw}. 
\end{proof}

\section{Conclusions}

Our results indicate that \DC\ is significantly harder than \textsc{Graph
Coloring}, even on classes where the latter is easily in P. Though we have
completely characterized the complexity of the problem on split and chordal
graphs, its tractability on interval and proper interval graphs remains an
interesting open problem as already posed in \cite{HavetKS09}.

Beyond this, the results of this paper point to several potential future
directions. First, the algorithms we have given for cographs are both XP
parameterized by $\C$ or $\D$. Is it possible to obtain FPT algorithms? On a
related question, is it possible to obtain a faster sub-exponential time
algorithm for \DC\ on cographs? Second, is it possible to find other natural
classes of graphs, beyond trivially perfect graphs, which are structured enough
to make \DC\ tractable? Finally, in this paper we have not considered the
question of approximation algorithms. Though in general \DC\ is likely to be
quite hard to approximate (as a consequence of the hardness of \textsc{Graph
Coloring}), it seems promising to also investigate this question in classes
where \textsc{Graph Coloring} is in P.

\bibliographystyle{abbrvnat}
\bibliography{defective}

\end{document}